\documentclass[conference]{IEEEtran}
\IEEEoverridecommandlockouts
\usepackage{cite}
\usepackage{amsmath,amssymb,amsfonts}
\usepackage{algorithmic}
\usepackage{graphicx}
\usepackage{textcomp}
\usepackage{xcolor, amsthm}
\usepackage[norelsize, linesnumbered, ruled, lined, boxed, commentsnumbered]{algorithm2e}

\newcommand{\matW}{\mathbf{W}}
\newcommand{\matV}{\mathbf{V}}

\newcommand{\matA}{\mathbf{A}}
\newcommand{\matI}{\mathbf{I}}
\newcommand{\matC}{\mathbf{C}}
\newcommand{\matB}{\mathbf{B}}
\newcommand{\matS}{\mathbf{S}}
\newcommand{\matD}{\mathbf{D}}
\newcommand{\matP}{\mathbf{P}}

\newcommand{\vecx}{\mathbf{x}}
\newcommand{\vecy}{\mathbf{y}}
\newcommand{\veca}{\mathbf{a}}
\newcommand{\vecr}{\mathbf{r}}

\usepackage[utf8]{inputenc}
\usepackage[english]{babel}

\newtheorem{prop}{Proposition}
\newtheorem{definition}{Definition}
\newtheorem{remark}{Remark}

\SetKwInput{KwInput}{Input}                
\SetKwInput{KwOutput}{Output}              

\DeclareMathOperator*{\argmin}{arg\,min}
\DeclareMathOperator*{\argmax}{arg\,max}

\def\BibTeX{{\rm B\kern-.05em{\sc i\kern-.025em b}\kern-.08em
    T\kern-.1667em\lower.7ex\hbox{E}\kern-.125emX}}
\begin{document}

\title{Path-Aware OMP Algorithms for Provenance Recovery in Wireless Networks}

\author{\IEEEauthorblockN{Shilpi Mishra$^{\dagger}$, J. Harshan$^{\dagger}$ and Ranjitha Prasad$^{*}$}
\IEEEauthorblockA{$^{\dagger}$Bharti School of Telecommunication Technology and Management, IIT Delhi, India\\
$^{*}$Department of Electronics and Communication Engineering, IIIT Delhi, India
}
}

\maketitle

\begin{abstract}
Low-latency provenance embedding methods have received traction in wireless networks for their ability to track the footprint of information flow. One such known method is based on Bloom filters wherein the nodes that forward the packets appropriately choose a certain number of hash functions to  embed their signatures in a shared space in the packet. Although Bloom filter methods can achieve the required accuracy level in provenance recovery, they are known to incur higher processing delay since higher number of hash functions are needed to meet the accuracy level. Motivated by this behaviour, we identify a regime of delay-constraints within which new provenance embedding methods must be proposed as Bloom filter methods are no longer applicable. To fill this research gap, we present network-coded edge embedding (NCEE) protocols that facilitate low-latency routing of packets in wireless network applications. First, we show that the problem of designing provenance recovery methods for the NCEE protocol is equivalent to the celebrated problem of compressed sensing, however, with additional constraints of path formation on the solution. Subsequently, we present a family of path-aware orthogonal matching pursuit algorithms that jointly incorporates the sparsity and path constraints. Through extensive simulation results, we show that our algorithms enjoy low-complexity implementation, and also improve the path recovery performance when compared to path-agnostic counterparts.
\end{abstract}


\section{Introduction}
\label{sec:intro}
Wireless multi-hop networks have found extensive applications in scenarios where direct communication between a source and a destination is limited due to multi-path fading, short-range coverage area, etc. Traditionally, multi-hop networks have been extensively studied from the perspective of optimizing end-to-end reliability \cite{chen2018multiple}. However, concurrent efforts on designing these networks from the viewpoint of security threats are also known  \cite{shi2004designing}. Within the latter class of efforts, it is well-known that multi-hop networks must be designed to track the footprint of information flow as it helps in network diagnostics as well as in detecting security threats on the packets. Specifically, in this context, the footprint of information flow refers to the knowledge of various processes that modify the packet en-route to the destination \cite{keller2012your}. Formally, the information related to the identity and the order of the nodes that modify and transmit the packets is called the \emph{provenance}, the method in which this information is embedded in the packet is called \emph{provenance embedding}, and finally, the process with which the provenance information is recovered at the destination is called \emph{provenance recovery} \cite{lim2010provenance}. 

\begin{figure}
\begin{center}
\includegraphics[scale=0.6]{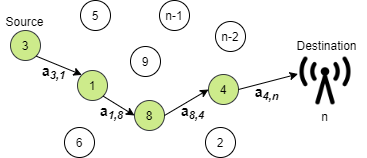}
\vspace{-0.4cm}
\end{center}
\caption{Multi-hop network model consisting of a source that communicates data to the destination node through intermediate relay nodes. The nodes use the network-coded edge embedding protocol to assist the destination in recovering the provenance with low-latency requirements.}
\label{fig:6}
\end{figure}

In the context of vehicular networks \cite{SVN}, which is a use-case of multi-hop wireless networks, a key requirement for designing provenance embedding methods is the end-to-end latency constraint imposed by the underlying application \cite{she2019ultra}. Although the process of decoding and forwarding the packets result in additional overhead, the process of provenance embedding also consumes time, thereby adding substantial end-to-end delay. As a result, low-complexity provenance embedding methods must be designed so that negligible delay is incurred on the packets. While low-complexity provenance embedding is important to meet the delay constraints during the packet routing phase, low-complexity provenance recovery is equally important. For instance, the destination could be a road side unit (RSU), and it may have limited resources to recover the provenance information from the packet. Thus, one needs to jointly design provenance embedding and recovery methods so that the required accuracy and latency constraints are met.

\begin{figure}
\begin{center}
\includegraphics[scale=0.41]{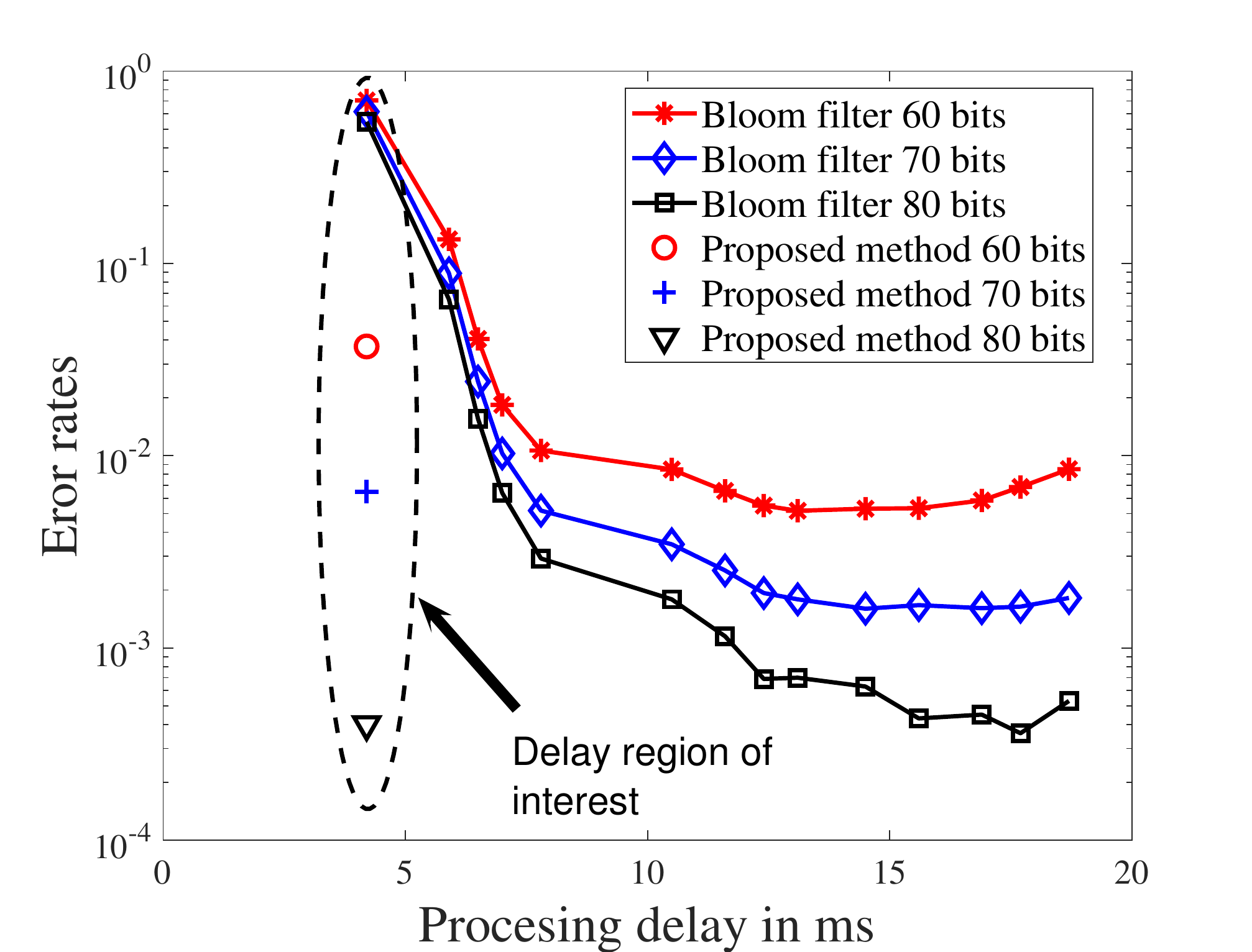}
\vspace{-0.4cm}
\end{center}
\caption{Simulation results depicting the trade-off between the processing delay and the error rates when implementing edge embedding using Bloom filters on a network with $15$ nodes and $4$-hop communication. XBee networks were used to compute the processing delays for the above plots.}
\label{fig_motivation}
\end{figure}

\subsection{Contributions and Novelty}
 
In this work, we address the design of provenance embedding and provenance recovery methods with two constraints. The first constraint is to design provenance embedding with low delay-overheads during packets' journey. The second constraint is to facilitate provenance recovery at the destination without the knowledge of the network topology owing to mobility of nodes. To jointly handle these constraints, we propose a new provenance embedding method called the network-coded edge embedding (NCEE) protocol. In this protocol, a unique identity is assigned to the edges between every pair of nodes in the network, and these identities are pre-shared with the nodes. Subsequently, every node that forwards the packet uses the identity of the edge between itself and the next node (as shown in Figure \ref{fig_motivation}), and algebraically adds it to the existing provenance portion of the packet. This way, the destination, upon receiving a linear combination of a subset of the signatures in the network, attempts to recover the path traversed by the packet by using the signatures of all the edges. While algebraic addition provides the low-latency feature during packets' journey, the use of edge identities as signatures assists the destination in resolving the order of the nodes on the path.\footnote{As an alternative approach, every node could embed its signature instead of the signature allotted to the edge with its next node. However, in this case, since the topology of the network may not be known to the destination, recovering the identity of the nodes \cite{liu2015path} does not reveal the information on the order of nodes that forwarded the packet.} Our specific contributions are listed below:

\noindent 1) We show that the problem of designing provenance recovery methods for the NCEE protocol is equivalent to the celebrated problem of compressed sensing, however, with the additional constraint of path formation on the sparse solution. As a consequence, we derive sufficient conditions on the solution by incorporating the path constraint (see Section \ref{sec:path_constraint}).

\noindent 2) We propose a family of low-complexity orthogonal matching pursuit (OMP)-based algorithms for recovering the provenance of the NCEE protocol. Specifically, we propose a list based OMP algorithm, referred to as the L-OMP algorithm, wherein every iteration of the traditional OMP algorithm collects more than one candidate vectors thereby constructing a list of sparse solutions of varying residual values. This way, the derived path constraint on the sparse solution is applied on the candidates in the list before picking a candidate path with minimum residue as the solution. We show that L-OMP algorithms improve the accuracy of provenance recovery (see Section \ref{sec:lomp}). We further enhance the L-OMP algorithms to propose path-aware list based OMP algorithm, referred to as the PL-OMP algorithm, which enforces the path constraint on the solution at an intermediate step of the L-OMP algorithm. We show that the PL-OMP variant provides substantial benefits in time-complexity as well as in the accuracy of recovery performance (see Section \ref{sec:plomp}).

\noindent 3) Since generalized OMP (gOMP) \cite{gOMP} is a promising extension of OMP, we also extend the list based ideas by proposing L-gOMP algorithms to further reap benefits in the accuracy of provenance recovery. Extensive simulation results are also presented to demonstrate that the proposed OMP and gOMP algorithms showcase superior performance when compared to their path-agnostic counterparts (see Section \ref{sec:sims}). 

In terms of novelty, edge embedding methods are already known in the context of Bloom filters \cite{HTDSC}. However, as exemplified in Fig. \ref{fig_motivation}, Bloom filter based methods exhibit a trade-off between the processing delay and the error rates. This behaviour is attributed to the fact that lower error rates are offered when the number of hash functions used by each node is bounded away from one. As shown in Fig. \ref{fig_motivation}, higher number of hash functions leads to higher processing delay at each node. This implies that when vehicular network applications require delay constraints equivalent to that of one or two hash functions, Bloom filters yield sub-optimal error rates, thereby creating opportunities for new provenance embedding methods. The NCEE protocols fill this gap since the delay incurred when computing linear combinations of edge identities falls in the above pointed delay regime.

\section{Network Model for Provenance Recovery}

We consider a wireless network of $n$ nodes, as shown in Fig.~\ref{fig:6}, wherein one of the nodes is the destination whereas the other nodes have data to communicate with the destination. In the context of vehicular networks, these $n-1$ mobile nodes represent vehicles whereas the destination is the RSU. Assuming that the network topology remains static for a given time interval, referred to as the coherence interval, we represent the network using a graph, denoted by $G(N, E)$, where $N = \{1, 2, \ldots, n\}$ is the set of vertices and $E \subset \{(i, j)~|~ i, j \in N, i \neq j\}$ is the set of edges. Owing to changing topology across coherence intervals, we assume that the destination does not know $G(N, E)$. Suppose node $i_{1}$, intends to communicate a message to the RSU in multi-hop fashion through the relay nodes $i_{2}$, $\ldots$, $i_{h}$, where $h$ denotes the hop-length.\footnote{This work does not concern designing routing algorithms. Instead, assuming the existence of an underlying routing algorithm, we design provenance embedding methods and the corresponding provenance recovery algorithms.} Due to ad-hoc routing of the packets, the path taken by the packets may not be known to the destination. Therefore, it is imperative for the participating nodes to not only forward the packet to the RSU, but also assist the RSU in learning the path on a packet-to-packet basis. In the next section, we propose the NCEE protocol on the above network model.

\subsection{Network-Coded Edge Embedding Protocol}

For a given node $i$ where $i \in N$, the set of all possible edge identities is denoted as $\{\veca_{i, j} \in \mathbb{R}^{m \times 1} ~|~ j \neq i\}$, where $\veca_{i, j}$ is the $m$-length signature of the directed edge from node $i$ to $j$. Formally, let the path chosen by the packet before reaching the destination be $i_{1}$, $i_{2}$, $\ldots$, $i_{h}$. To start the embedding process, the source inserts $\mathbf{a}_{i_{1}, i_{2}}$ in the provenance portion of the packet before forwarding it to node $i_{2}$. Subsequently, $i_{2}$ extracts the provenance, and adds  $\mathbf{a}_{i_{2}, i_{3}}$ to the provenance, thereby forwarding $\mathbf{a}_{i_{1}, i_{2}} + \mathbf{a}_{i_{2}, i_{3}}$. We assume that the packet also hosts a counter, which is incremented at each node on the path. This process continues at all the nodes on the path, and hence, the received provenance at the destination is
\begin{equation}
\label{eq:edge}
\vecy_{E} = \sum_{r = 1}^{h} \veca_{i_{r}, i_{r+1}},
\end{equation}
such that $i_{h + 1} = n$. For the provenance recovery process, the destination first obtains the $m \times (n-1)^2$ matrix $\mathbf{A}_{E}$ by juxtaposing all the edge identities as
\begin{align}
    \matA_{E} = \left[\matA_{E_1}, \matA_{E_2},\hdots,\matA_{E_{n-1}}\right],
    \label{eq:Astructure}
\end{align}
where $\matA_{E_i} \in \mathbb{R}^{m \times (n-1)}$ is the  set of $n-1$ columns corresponding to the edge identities between node $i$ and all other nodes, arranged in the increasing order of their indices. In particular, we have $\matA_{E_i} = [\mathbf{a}_{i, 1} ~\mathbf{a}_{i, 2} \ldots, \mathbf{a}_{i, i-1}, \mathbf{a}_{i, i+1}, \ldots, \mathbf{a}_{i, n}],$ such that $\mathbf{a}_{i, j}$ is the signature as defined earlier. Note that the number of columns in $\mathbf{A}_{E}$ is $(n-1)^{2}$ instead of $n(n-1)$ since the destination discards the edges originating from itself to the other nodes. Subsequently, the destination reformulates the problem of provenance recovery as a problem of sparse recovery of $\hat{\mathbf{x}} \in \{0, 1\}^{(n-1)^2 \times 1}$ such that (i) $\vecy_{E} = \mathbf{A}_{E}\hat{\vecx}$, (ii) the sparsity of $\hat{\vecx}$ is $h$, and (iii) the edge identities corresponding to the non-zero entries of $\hat{\vecx}$ form a path in the network. Note that $h$ is known to the destination from the counter maintained in the packet. While a straightforward method to recover sparse $\hat{\vecx}$ is to employ a standard sparse-recovery algorithms, such methods do not necessarily yield the correct solution since $\hat{\vecx}$ may not be a path. Thus, it is important to design \emph{path-aware} sparse recovery algorithms. 

In the sequel, we derive path constraints on $\hat{\vecx}$, and subsequently, use the path constraints to propose OMP-based path-aware algorithms for provenance recovery.

\section{Path Constraints for Edge Embedding}
\label{sec:path_constraint}

From \eqref{eq:edge}, it is clear that the provenance received at the destination can be represented as $\vecy_{E} = \matA_{E}\vecx$, where $\mathbf{A}_{E} \in \mathbb{R}^{m \times (n-1)^{2}}$ is as given in \eqref{eq:Astructure}, and $\vecx \in \{0, 1\}^{(n-1)^2 \times 1}$ is the $h$-sparse vector that represents the path taken by the packet. Owing to the structure of $\mathbf{A}_{E}$, $\mathbf{x}$ also gets arranged as
\begin{align}
    \mathbf{x} = \left[\mathbf{x}_{E_1}^T,\mathbf{x}_{E_2}^T,\hdots,\mathbf{x}_{E_n}^T\right]^T,
    \label{eq:xstructure}
\end{align}
where $\mathbf{x}_{E_i} \in \{0, 1\}^{n-1 \times 1}$ represents the set of $n-1$ entries corresponding to edges between node $i$ and all other nodes as $\mathbf{x}_{E_i} = [x_{i, 1} ~x_{i, 2} \ldots, x_{i, i-1}, x_{i, i+1}, \ldots, x_{i, n}]^T.$ Given that the system of linear equations in \eqref{eq:edge} does not allow any arbitrary $h$-sparse vector in $\{0, 1\}^{(n-1)^2 \times 1}$ as a solution, we define $\mathcal{P}_{h} \subset \{0, 1\}^{(n-1)^2 \times 1}$ to be the set of all $h$-sparse vectors that form a $h$-length path from a source to the destination. Using this formulation, provenance recovery is possible by solving
\begin{align}
    \hat{\vecx} = \argmin_{\vecx \in \mathcal{P}_{h}} ||{\vecy - \matA_{E}\vecx}||,
    \label{eq:traditionalCombinatorial}
\end{align}
where $||\cdot||$ represents the $\ell_2$ norm. Note that imposing the constraint $\vecx \in \mathcal{P}_h$ requires exhaustive search, and the complexity of solving such a problem increases exponentially with the size of the path. To alleviate this problem, we convert the traditional combinatorial optimization problem in \eqref{eq:traditionalCombinatorial} into a continuous program as given by
\begin{align}
    \hat{\vecx} = \argmin_{\vecx \in \mathbb{R}^{(n-1)^2}} ||{\vecy - \matA_{E}\vecx}|| \quad\quad \nonumber \\ \textnormal{subject to} 
    \quad \{g_{\theta}(\vecx) \leq 0, 1\leq \theta \leq \Theta\},
    \label{eq:traditionalsmooth}
\end{align}
where the set of $\Theta$ functions $\{g_{\theta}(\cdot)\}$, capture the path constraint $\mathbf{x} \in \mathcal{P}_{h}$. Although the problems in \eqref{eq:traditionalCombinatorial} and \eqref{eq:traditionalsmooth} are equivalent, the continuous program in \eqref{eq:traditionalsmooth} eliminates the need for an exhaustive search over $\mathcal{P}_{h}$, and instead, allows us to leverage standard numerical algorithms for constrained problems leading to ease of implementation.

Now we derive the set of constraints $\{g_{\theta}(\vecx) \leq 0, 1 \leq \theta \leq \Theta\}$, that is instrumental in establishing that $\vecx \in \mathcal{P}_{h}$. For a given $\vecx \in \{0,1\}^{(n-1)^2}$, let $G_{\vecx}$ denote the graph formed by the edges corresponding to the non-zero entries of $\mathbf{x}$ and let $\mathbf{W}_{x} \in \{0, 1\}^{n \times n}$ denote the adjacency matrix corresponding to $G_{\vecx}$. If $|\mathbf{x}|_{1} = h$, then $\mathbf{W}_{x}$ contains $h$ non-zero off-diagonal entries. The diagonal entries and entries of the last row of $\mathbf{W}_{\mathbf{x}}$ are implicitly zeros since there are no self-loops in the graph, and edges originating from node $n$ are discarded. We transform the $(n-1)^2$-length $\vecx$ in \eqref{eq:xstructure} to an $n^{2}$-length $\bar{\vecx} \in \{0, 1\}^{n^{2} \times 1}$,
\begin{align}
    \bar{\vecx} = \left[[\bar{\vecx}_{E_1}]^T,[\bar{\vecx}_{E_2}]^T,\hdots,[\bar{\vecx}_{E_n}]^T\right]^T,
    \label{eq:x_new_structure}
\end{align}
where $[\bar{\vecx}_{E_n}]^T = [0,~0,~\ldots,~0]$ is an $n$-length zero vector, and $\bar{\vecx}_{E_i} \in \{0, 1\}^{n}$, for $1 \leq i \leq n-1$, is given by $\bar{\vecx}_{E_i} = [x_{i, 1} ~x_{i, 2} \ldots, x_{i, i-1}, 0, x_{i, i+1}, \ldots, x_{i, n}]^T$. Note that $\vecx_{E_i}$ is converted to $\bar{\vecx}_{E_i}$ by inserting $0$ at the location $x_{i,i}$. Note that the adjacency matrices corresponding to $\vecx$ and $\bar{\vecx}$ are identical. Using $\bar{\vecx}$, we construct $\mathbf{W}_{x}$ for any given $\mathbf{x} \in \{0, 1\}^{(n-1)^2 \times 1}$, in the following proposition.

\begin{prop}
The adjacency matrix $\matW_{x}$ corresponding to the subgraph $G_\vecx$ of $\mathbf{x}$ is $\matW_{x} = \matB\matV\matC$, where $\matV = \textnormal{diag}(\bar{\vecx})$, $\matB \in \mathbb{R}^{n \times n^2}$ is given by $\matB = \matI_n \otimes \mathbf{1}$, and $\matC \in \mathbb{R}^{n^2 \times n}$ is given by $\matC = \mathbf{1}^T \otimes \matI_n$, where $\mathbf{1}$ is an all one column vector of length $n$.   
\label{prop1}
\end{prop}

\begin{proof}
Let $\mathbf{V} = \textnormal{diag}(\bar{\vecx})$ be the diagonal matrix with diagonal entries given by $\bar{\vecx}$. We observe that $\mathbf{V}$ is a block-diagonal matrix containing $n$ blocks such that the $i$-th block-diagonal matrix, for $1 \leq i \leq n$, contains the identity of the edges originating from node $i$. Therefore, if $\mathbf{V}(d, d) = 1$, for $v \in 1 \leq v \leq n^{2}$, then the $\lceil{\frac{d}{n}}\rceil$th row and the $(d-(n(\lceil{\frac{d}{n}}\rceil-1)))$-th column entry of $\mathbf{W}_{x}$ must be one. To achieve this, we multiply $\matV$ by $\matB \in \mathbb{R}^{n \times n^2}$ from the left where $\matB = \matI_n \otimes \mathbf{1}$. Furthermore, in order to convert $\matB\matV$ into a adjacency matrix, the edge information contained in these $n$ blocks are brought together as an $n \times n$ matrix by multiplying $\matB\matV$ by $\matC \in \mathbb{R}^{n^2 \times n}$ from the right where $\matC = \mathbf{1}^T \otimes \matI_n$. This completes the proof.
\end{proof}

\begin{prop}
A $h$-sparse vector $\mathbf{x} \in \{0, 1\}^{(n-1)^2 \times 1}$ belongs to $\mathcal{P}_{h}$ if there exists $i \in \mathbb{N}$ such that $\mathbf{1}_i^T\matW_{x}^h\mathbf{1}_n - 1 = 0$ and $\mathbf{1}_i^T\matW_{x}^t\mathbf{1}_n = 0, \mbox{ for } 1 \leq t \leq h-1$ such that $\mathbf{1}_i$ represents a vector with one at the $i$-th entry and zeros at other places. 
\end{prop}
\begin{proof}
From Proposition~\ref{prop1}, we can obtain the adjacency matrix $\matW_{x}$ corresponding to $G_x$. We also know that the destination is node $n$. It is well known that if the $(i,n)$-th entry in $\matW_{x}^h$ is one, and the $(i,n)$-th entry in $\matW_{x}^l$ is $0$ for $1 \leq l < h$, then there exists a $h$-length path between node $i$ and node $n$ \cite{Rosen}. In the statement of the theorem, we used the vectors $\mathbf{1}_i^T$ and $\mathbf{1}_n$ to fetch the $(i,n)$-th entry of $\mathbf{W}_{x}$ using $\mathbf{1}_i^T\matW^l\mathbf{1}_n$, where $\mathbf{1}_i$ is an $n$ length column vector that has one at the $i$-th position and zeros at all other positions. This completes the proof.  
\end{proof}

The above proposition leads to the set of conditions specified in \eqref{eq:traditionalsmooth}. One of the popular methods to solve \eqref{eq:traditionalsmooth} is to use the basis pursuit based optimization framework \cite{471413}, which in turn necessitates using high-complexity optimization solvers. However, from an implementation viewpoint, a well-known low-complexity greedy technique to solve the sparse recovery problem is the OMP algorithm. Therefore, in the next section, we modify the vanilla OMP framework by integrating the constraints $\mathbf{1}_i^T\matW_{x}^h\mathbf{1}_n - 1 = 0$ and $\mathbf{1}_i^T\matW_{x}^t\mathbf{1}_n = 0, \mbox{ for } 1 \leq t \leq h-1$, so that the solution $\hat{\vecx}$ is a path.

\section{Path-Aware OMP Algorithms}

In this section, we present a family of OMP based algorithms that are designed to increase the likelihood of recovering a path by using the path-constraints of the previous section. Towards that direction, the following definition is important. 

\begin{definition}
\label{beta_max}
Given a set $\mathcal{M}$ containing a finite set of numbers in $\mathbb{R}^{+}$, the notation $\mathop{\beta\mbox{-}\textnormal{arg\,max}}$ $\mathcal{M}$, for $\beta \in \mathbb{N}$, represents $\beta$-th largest element of $\mathcal{M}$.  
\end{definition}

For a given $h, L \in \mathbb{N}$, let $\Gamma = (\alpha_{1}, \alpha_{2}, \ldots, \alpha_{h})$ such that $\alpha_{j} \in [L]$, for any $1 \leq j \leq h$. Furthermore, for any subset $S \subseteq  [(i,j)~|~ i \in N, j \in N, i \neq j]$, let $\matA_{S}$ represent a sub-matrix of $\matA_E$, formed by the columns indexed by $S$. Using the above definitions, we define a variant of the well known OMP algorithm, henceforth, referred to as $\Gamma$-OMP algorithm.

\subsection{$\Gamma$-OMP Algorithm}
\label{gamma_OMP}

For a given $\Gamma = (\alpha_{1}, \alpha_{2}, \ldots, \alpha_{h})$ defined above, the steps of the proposed $\Gamma$-OMP algorithm are given below.
 
\begin{itemize}
    \item Step 1: Let $\vecr_{0} = \vecy_{E}$, $S = \{\}$, and $\eta = 1$.
    \item Step 2: Find the column with index $\veca_{i_{\eta},j_{\eta}}$ of $\matA$ that maximizes the following:
    \begin{align}
    \label{eq:max_projection}
      (i_{\eta},j_{\eta})  = \mathop{\alpha_\eta\mbox{-}\textnormal{arg\,max}}\limits_{p, q \in N, p\neq q} \veca_{(p,q)}^T\vecr_{\eta - 1},
    \end{align}
    wherein the notion of $\alpha_{\eta}\mbox{-}\argmax$ follows from Definition \ref{beta_max}. Update the set $S$ as $S \leftarrow S \cup (i_{\eta},j_{\eta})$. 
\item Step 3: Compute $\matP_{S} = \matA_{S}(\matA_{S}^T\matA_{S})^{-1}\matA_{S}^T$. Update the residual $\vecr_{\eta} = (\matI -\matP_{S})\vecy_E$.
\item Step 4: If $\eta < h$, update $\eta \leftarrow \eta+1$ and return to Step 2, otherwise, exit the algorithm.
\end{itemize}
At the output $S = \{(i_{1}, j_{1}), (i_{2}, j_{2}), \ldots, (i_{h}, j_{h})\}$ contains the list of $h$ edges. A distinguishing feature of this algorithm is \eqref{eq:max_projection}, wherein $\alpha_{\eta}\mbox{-}\argmax$ operator is specific for a given iteration ($\eta$) of the algorithm. Note that the above algorithm collapses to the vanilla OMP algorithm when $\Gamma = [1, 1, \ldots, 1]$. 

The $\Gamma$-OMP algorithm fails to solve \eqref{eq:traditionalsmooth} as it does not incorporate the path constraint. While the path constraint check can be applied on its final solution as a retrofit, such a step only helps in verifying whether $\vecx \in \mathcal{P}_h$, thereby not altering the error rates in provenance recovery.  

\subsection{List Based OMP Algorithm (L-OMP)}
\label{sec:lomp}

In this section, we extend the vanilla OMP algorithm such that each iteration generates a list of $L$ columns of $\mathbf{A}_{E}$ that maximizes the projection w.r.t. the residue in \eqref{eq:max_projection}. In particular, for every candidate in the list generated at the $\eta$-th iteration, a new list of $L$ columns is generated at the $(\eta+1)$-th iteration, thereby leading to a list of $L^{h}$ solutions after $h$ iterations. Among these $L^{h}$ candidates, we select $\vecx \in \mathcal{P}_h$ if it exists, and then specifically choose a solution that has the least residue among them. The L-OMP algorithm can be stated elegantly using $\Gamma$-OMP in an iterative manner as follows: 
\begin{itemize}
    \item Step~1: Set $r \leftarrow 1$. Construct an ordered set $\mathcal{S} = \{(\alpha_{1}, \alpha_{2}, \ldots, \alpha_{h}) ~|~ \alpha_{j} \in [L], \forall j\}$.  Set $\mathcal{L} \leftarrow \{\}$ to store the list of OMP solutions. 
    \item Step~2: $\Gamma \leftarrow \mathcal{S}(r)$, where $\mathcal{S}(r)$ is the $r$-th entry of $\mathcal{S}$.   
    \item Step 3: Invoke $\Gamma$-OMP algorithm, and append its solution to $\mathcal{L}$.
    \item Step 4: If $r < L^{h}$, then update $r \leftarrow r + 1$, and go to Step 2, otherwise, go to Step 5.
    \item Step 5: Use path-constraint check on the entries of $\mathcal{L}$ to generate a shorter-list $\bar{\mathcal{L}}$ comprising candidates that satisfy the path-constraint. 
    \item Step 6: Among the candidates in $\bar{\mathcal{L}}$, pick the solution with minimum residue, i.e., $\hat{\mathbf{x}} =  \arg \min_{\mathbf{x} \in \bar{\mathcal{L}}}||\mathbf{y}_{E} - \mathbf{A}_{E}x||$.
\end{itemize}
In contrast to vanilla OMP, the L-OMP algorithm increases the likelihood of recovering a solution in $\mathcal{P}_{h}$.  

\subsection{Path-aware List-based OMP (PL-OMP) Algorithm}
\label{sec:plomp}

Although the L-OMP algorithm increases the likelihood of recovering a path by generating a list of solutions, it does not introduce fundamental changes to incorporate path awareness in every iteration. To impose a path constraint during the $\eta$-th iteration of the OMP algorithm, it is clear that the edge $(i_{\eta}, j_{\eta})$ obtained from \eqref{eq:max_projection} must be such that $i_{\eta} \notin \mathcal{I}_{\eta - 1} \triangleq \{i_{1}, i_{2}, \ldots, i_{\eta - 1}\}$ and $j_{\eta - 1} \notin \mathcal{J}_{\eta - 1} \triangleq \{j_{1}, j_{2}, \ldots, j_{\eta - 1}\}$. Therefore, we propose to replace \eqref{eq:max_projection} by
 \begin{align}
    \label{eq:max_projection_path_aware}
      (i_{\eta},j_{\eta})  = \mathop{\alpha_\eta\mbox{-}\textnormal{arg\,max}}\limits_{p, q \in N, p\neq q, p \notin \mathcal{I}_{\eta - 1}, q \notin \mathcal{J}_{\eta-1}} \veca_{(p,q)}^T\vecr_{\eta - 1}.
    \end{align}
In addition to the above change, we note that as $L$ and $h$ increase, the implementation complexity of the L-OMP algorithm also increases. To circumvent this problem, we propose a novel path-aware low-complexity variant of the L-OMP algorithm, with improved accuracy in some cases. The central idea of PL-OMP is to generate a shorter list (of size $L^{h-1}$) by executing the L-OMP algorithm for sparsity $h-1$, and then obtain the $h$-sparse solution by completing the path on those candidates that have one missing link in the solution. Formally, we use the following definition to introduce the missing-link solution.

\begin{definition}
Let $\mathbf{x} \in \{0, 1\}^{(n-1)^2 \times 1}$ represent a solution of the L-OMP algorithm at the end of the $(h-1)$-th iteration. Furthermore, let $G_{\mathbf{x}}$ denote the subgraph corresponding to $\mathbf{x}$. Then $\mathbf{x}$ is referred to as the missing-link solution if $G_{\mathbf{x}}$ comprises two disjoint paths of lengths $a \geq 0$ and $b \geq 0$ such that $a + b = h-1$.
\end{definition}

Once a missing-link solution is obtained, it is straightforward to observe that a path of hop-length $h$ can be formed by connecting the appropriate vertices of the two paths of $G_{\mathbf{x}}$. Note that this task can be achieved since the identity of the source node and the destination node are known. The following algorithm provides a method to check whether a candidate in the list (at the end of the $(h-1)$-th iteration) is a missing-link, or otherwise. To execute this check on any $\mathbf{x} \in \{0, 1\}^{(n-1)^2}$, we use the corresponding adjacency matrix $\mathbf{W}_{\mathbf{x}}$ obtained using Proposition \ref{prop1}. The algorithm for missing-link detection is given below. For the description of the algorithm, we assume that node $s$ is the source node, and $\mathbf{W}_{\mathbf{x}}$ is already computed for the candidate $\mathbf{x}$ under test.
\begin{itemize}
    \item Step 1: If $\sum_{j=1}^{n} \matW_{\mathbf{x}}(k,j) > 1$ OR 
    $\sum_{i=1}^{n} \matW_{\mathbf{x}}(i,k) > 1$, for any $k$ such that $1 \leq k \leq n$, then candidate $\mathbf{x}$ has more than one edges from a node and does not qualify as a missing link solution, else go to Step 2.
    \item Step 2: Let $\mathbf{w}_{s}^i$ and $\mathbf{w}_{d}^i$ denote $\matW^i_{\mathbf{x}}(s,:)^\intercal$ and $\matW^i_{\mathbf{x}}(:,n)$.
    Construct 
    \begin{equation}
    \label{S_matrix}
    \matS = [\mathbf{w}_{s}^1, \mathbf{w}_{s}^2 \ldots \mathbf{w}_{s}^{h-1}] \in \mathbb{R}^{n \times h-1}
    \end{equation}
    \begin{equation}
    \label{D_matrix}
    \matD = [\mathbf{w}_{d}^1, \mathbf{w}_{d}^2 \ldots \mathbf{w}_{d}^{h-1}]  \in \mathbb{R}^{n \times h-1}.
    \end{equation}
    \item Step 3: Initialise: 
    $\mathbf{q} = [1,~0,~\ldots,~0]^{1 \times h-1}$, $a = 0, b = 0$, $node_a = s, node_b = n$; where $node_a$ and $node_b $ are the nodes connected to node s and node n, respectively. The variables $a$ and $b$ are used to store  path lengths from node s to $node_a$ and node n to $node_b$, respectively.
    \item Step 4: To find $a$ and $node_a$, call the function [$a$,$~node_a$] = \texttt{FINDLENGTH}($\mathbf{q},\mathbf{S}$).

    \item Step 5: To find $b$ and $node_b$, call the function [$b$,$~node_b$] = \texttt{FINDLENGTH}($\mathbf{q},\mathbf{D}$).
    
    \item Step 6: If $a + b = h-1$, then candidate $\mathbf{x}$ is a missing-link solution, else it is not. In the former case, the missing link is the edge from node $node_{a}$ to node $node_{b}$. 
\end{itemize}

The \texttt{FINDLENGTH} function that is used to find the nodes connected to the source the destination is presented in Algorithm \ref{Algo}.

\begin{algorithm}
\DontPrintSemicolon 
  \KwInput{length = 0, $\mathbf{q} = [1, 0, \ldots, 0]$, $\mathbf{Q}$ (which is either \eqref{S_matrix} or \eqref{D_matrix})}
  \KwOutput{length, node}
  \For{iter = 1 to $h-1$} 
  {
  \If{$\mathbf{Q}(i, :) == \mathbf{q}$ for any $1 \leq i \leq n-1$}
  {
  length $\leftarrow$ length + 1\\
  node $\leftarrow i$\\
  Update $\mathbf{q}$ after right shifting by one\\
  }
  \Else
  {
  break\\
  }
  }
  \caption{\label{Algo} \texttt{FINDLENGTH} Algorithm}
\end{algorithm}

We now present the PL-OMP algorithm that executes the L-OMP algorithm up to $h-1$ iterations, and then obtains the solution by applying the above missing-link algorithm. 
\begin{itemize}
    \item Step~1: Set $r \leftarrow 1$. Construct an ordered set $\mathcal{S} = \{(\alpha_{1}, \alpha_{2}, \ldots, \alpha_{h-1}) ~|~ \alpha_{j} \in [L], \forall j\}$ . Let $\mathcal{L} \leftarrow \{\}$ to store the output of $\Gamma$-OMP algorithm for each $\Gamma \in \mathcal{S}$. 
    \item Step~2: $\Gamma \leftarrow \mathcal{S}(r)$, where $\mathcal{S}(r)$ is the $r$-th element of  $\mathcal{S}$.   
    \item Step~3: Call the $\Gamma$-OMP algorithm in Section \ref{gamma_OMP} with inputs $\mathbf{y}_{E}$ and $\mathbf{A}_{E}$, and append its solution to $\mathcal{L}$. Note that \eqref{eq:max_projection} must be replaced by \eqref{eq:max_projection_path_aware} in the $\Gamma$-OMP algorithm. 
    \item Step~4: If $r < L^{h-1}$, then update $r \leftarrow r + 1$, and go to Step 2, otherwise, go to Step 5.
    \item Step~5: Run the missing-link algorithm on the elements of $\mathcal{L}$ to generate  $\bar{\mathcal{L}} \subset \mathcal{L}$ that consists of candidate solutions that satisfy the missing-link constraint. 
    \item Step~6: Complete the path for candidate solutions in $\bar{\mathcal{L}}$, and choose one path in $\bar{\mathcal{L}}$ that minimizes the residue. 
\end{itemize}

The following proposition provides guarantees on the accuracy of PL-OMP with respect to L-OMP.

\begin{prop}
For a given $L$ and $h$, the error rate of PL-OMP algorithm is upper bounded by that of L-OMP algorithm 
\end{prop}
\begin{proof}
At the end of the $(h-1)$-th iteration, the list of candidates (of size $L^{h-1}$) generated by L-OMP and PL-OMP are identical. Among the candidates in this list, a candidate that does not satisfy the missing-link criterion does not appear in $\bar{\mathcal{L}}$ of L-OMP algorithm irrespective of the edge added by the $h$-th iteration. This implies that the missing-link algorithm does not reject the solution of L-OMP algorithm. Furthermore, among the candidates that satisfy the missing link constraint, the L-OMP algorithm may or may not add the missing-link as part of its $h$-th iteration. However, the PL-OMP algorithm adds the correct missing-link by brute-force, thereby providing an opportunity to improve the error-rates of PL-OMP over the L-OMP algorithm. This completes the proof.
\end{proof}

In the rest of this section, we show that the worst-case complexity of PL-OMP is \emph{lower} than that of L-OMP. Since the first $h-1$ iterations for L-OMP and PL-OMP are the same, we discuss the difference in complexity of the two schemes while adding the $h$-th edge to the solution. First, we note that PL-OMP skips the $h$-th iteration of the traditional OMP on each candidate of the list, thereby bypassing the complexity needed to compute the pseudo-inverse of $\mathbf{A}_{S}$. This leads to a complexity reduction of $O((n-1)^2m + mh + mh^{2} + h^{3})$ \cite{6333943}. In addition, PL-OMP obtains the solution by applying the missing-link check on a list of size $L^{h-1}$, whereas L-OMP obtains the solution by applying the path-constraint check on a list of size $L^{h}$. Thus, the missing-link algorithm in PL-OMP has a complexity of $O(nhL^{h-1} + (h-1)n^{3}L^{h-1})$, whereas L-OMP has a complexity of $O(hL^{h} + hn^{3}L^{h}) + O(L^{h-1}((n-1)^2m + mh + mh^{2} + h^{3}))$ after the $h$-th iteration. In Fig.~\ref{fig:5}, we plot the above two worst-case numbers as well as their difference as a function of $L$ and $h$. From the plots, it is clear that for the same input parameters, PL-OMP requires fewer operations as compared to L-OMP.

\begin{figure}
\begin{center}
\includegraphics[scale=0.46]{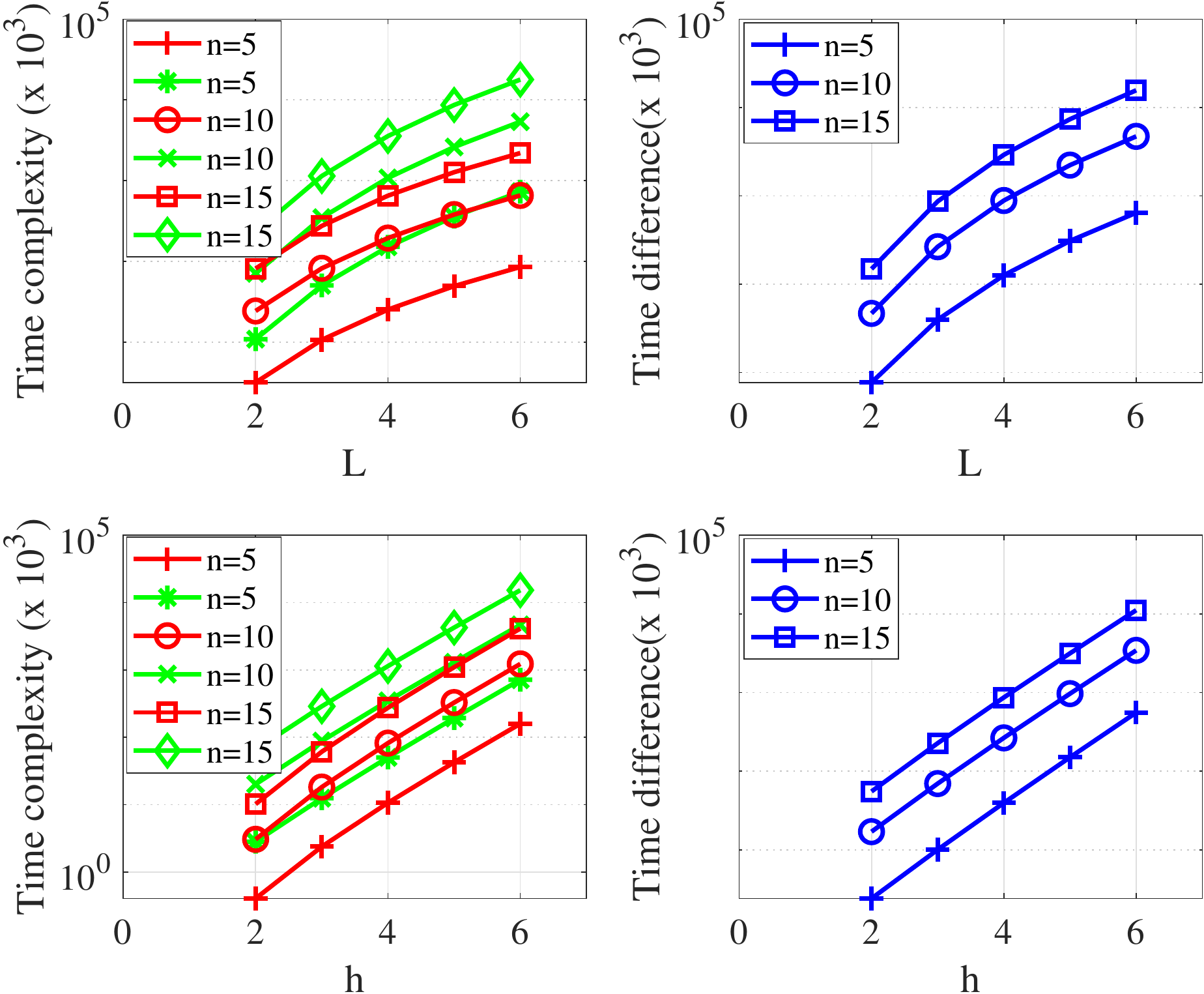}
\end{center}
\vspace{-0.3cm}
\caption{Time complexity of edge-embedding as a function of $L$ and $h$ for L-OMP and PL-OMP, with $m=8$. Red and green plots correspond to PL-OMP and L-OMP respectively. Here, $h$ and $L$ represent the hop-length and the list-size, respectively.}
\label{fig:5}
\end{figure}

\begin{figure}
\begin{center}
\includegraphics[scale=0.38]{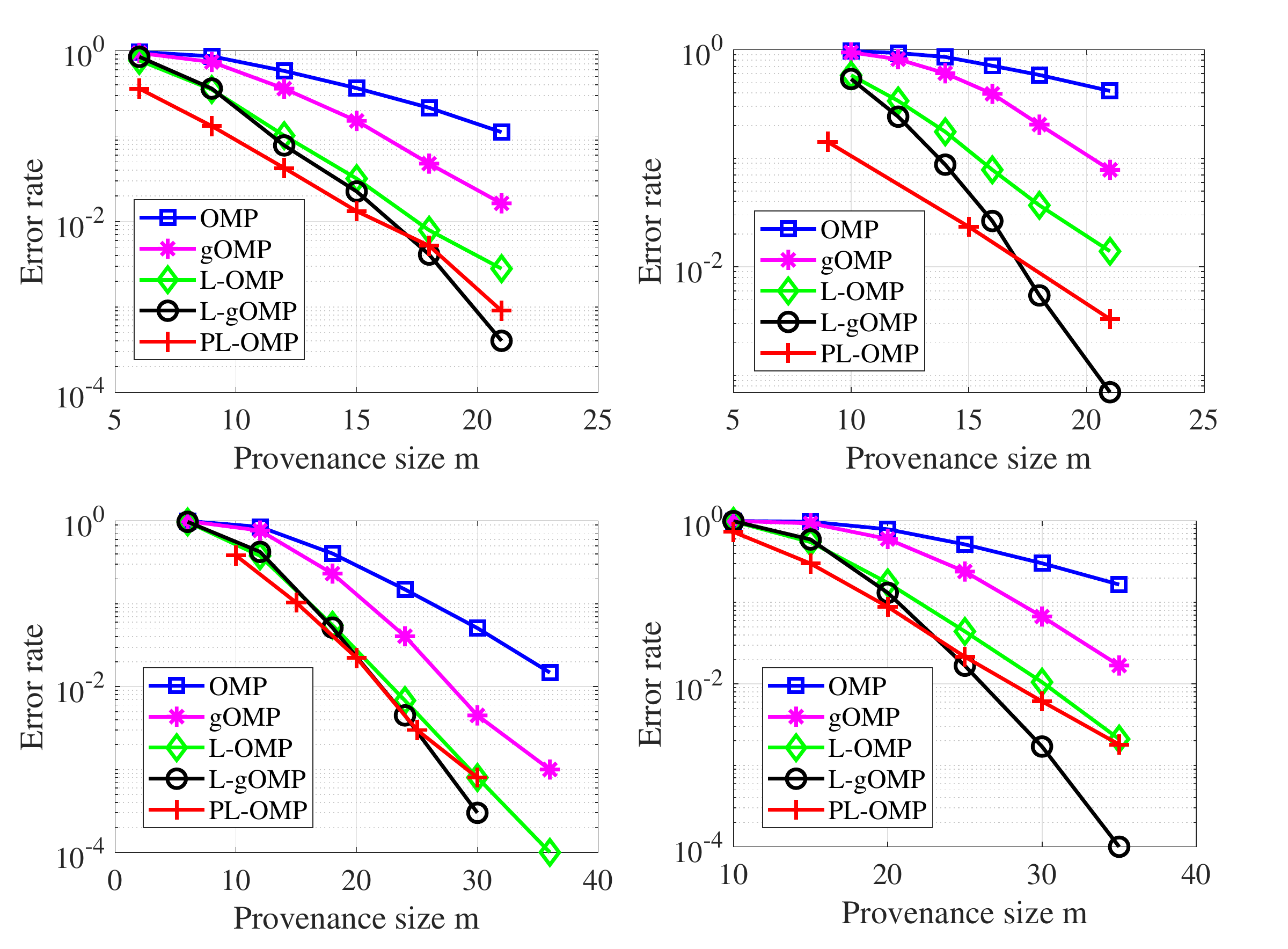}
\end{center}
\vspace{-0.5cm}
\caption{Error rate of OMP, gOMP and its variants as a function of $m$. Clockwise from top-left, the plots correspond to $(n,h)$ of $(6,3), (6,5), (9,5) \mbox{ and } (9,3)$.}
\label{fig:1}
\end{figure}

\begin{figure}
\begin{center}
\includegraphics[scale=0.36]{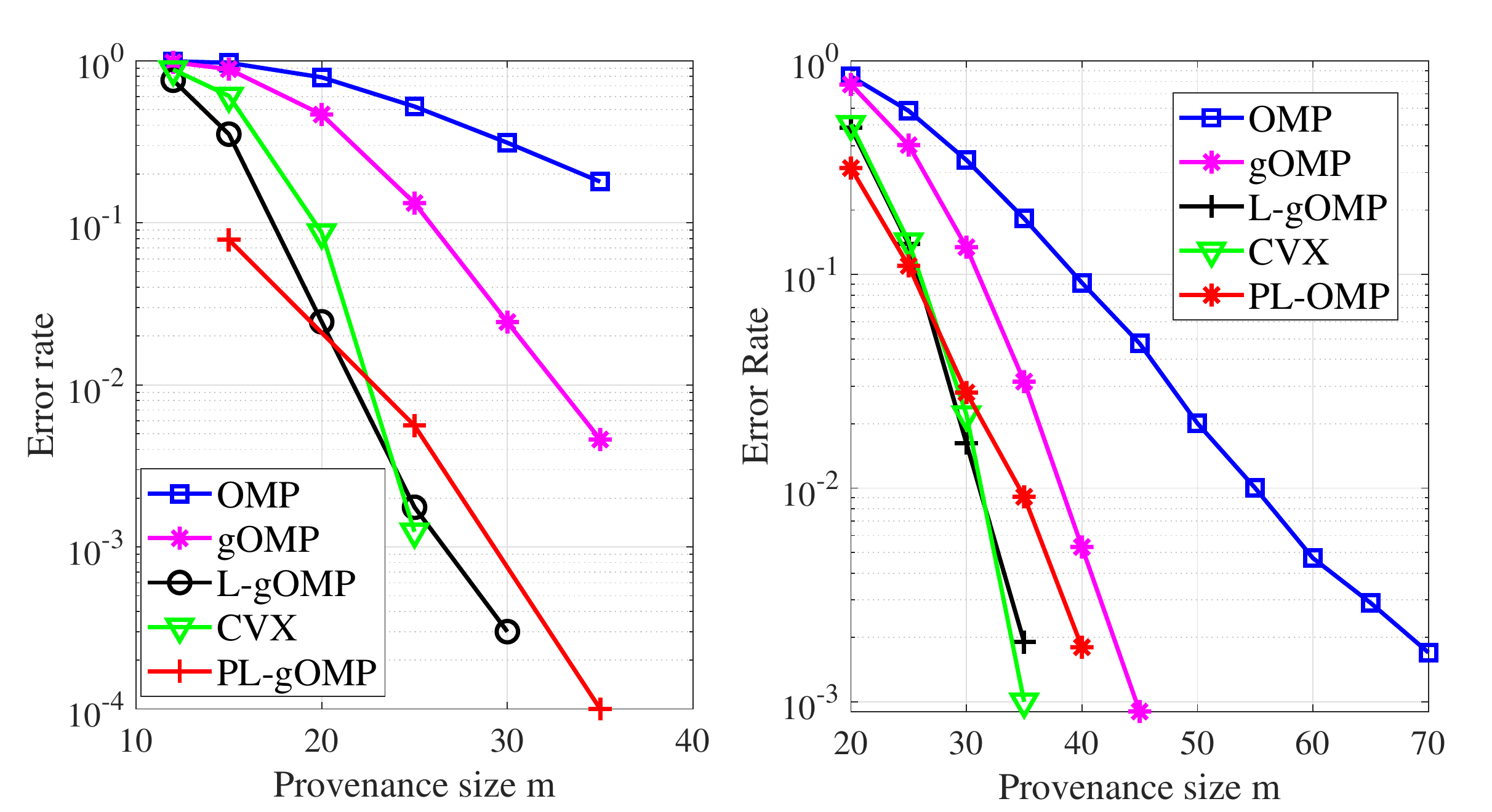}
\end{center}
\vspace{-0.5cm}
\caption{Error rate of CVX, OMP variants and gOMP variants as a function of $m$ with $n=7$, $h=6$ (left-side) and $n=15$, $h=4$ (right-side).}
\label{fig:3}
\end{figure}

\subsection{List based Generalized OMP Algorithms}

Generalized OMP algorithms (gOMP) \cite{gOMP} are well known extensions of the vanilla OMP algorithm for providing significant benefits in error rates. As its salient feature, every iteration in gOMP picks a group of $v$ columns, for $v > 1$, when maximising the projection w.r.t. the residue. This number is chosen based on the condition that the number of iterations is bounded by $min\{h, \frac{m}{v}\}$ (please see \cite[Table 1]{gOMP}). As a result, the number of iterations in gOMP can be fewer than that of the vanilla OMP algorithm. As an extension to L-OMP, we propose a list based gOMP (L-gOMP) in this section. To generate a list at every iteration of gOMP, we select the first $w$ columns of $\mathbf{A}_{E}$ from the projection operation such that $w > v$ (the identification step of \cite[Table 1]{gOMP}). Consequently, this provides a list of $w \choose v$ columns of $\mathbf{A}_{E}$ for computing the residual. Along the similar lines, for every batch of $v$ columns picked in the $j$-th iteration, a new list of $w \choose v$ columns of $\mathbf{A}_{E}$ is chosen in the $(j+1)$-th iteration. In this process, the L-gOMP produces a list of size less than or equal to $L^{\kappa}$ support sets where $L$ is $w \choose v$ and $\kappa$ is the total number of iterations used by the gOMP algorithm. Using this list, we execute the last step of gOMP (given by output step of \cite[Table 1]{gOMP})) on every candidate in the list to generate $L^{\kappa}$ sparse vectors. Finally, the proposed path-constraint criterion is applied on the list before choosing the path that provides minimum residue. 
\begin{remark}
Unlike the vanilla OMP algorithm, the number of columns fed to the last step of the gOMP can be more than the sparsity of $\mathbf{x}$. Also, each iteration of gOMP selects more than one column. Due to these properties, incorporating the missing-link algorithms to gOMP is not easy. In fact, even if we would like to terminate the gOMP algorithm one iteration early, we would need to complete the path by filling in two or more missing links depending on the value of $v$. Furthermore, since the support set to the last step can be more than sparsity, we would not be able to determine which specific columns among the set would be recovered in the end. Due to these observations, imposing path-aware constraints within the gOMP algorithm is a challenging task.
\end{remark}

\section{Simulation Results}
\label{sec:sims}

In this section, we demonstrate the efficacy of the proposed OMP-based algorithms when recovering the provenance in the NCEE protocol. For the simulation setup, we construct a complete graph for a given $n$, and then select an $h$-hop path between a source and the destination. For each $i, j \in N$ such that $i \neq j$, we obtain $\veca_{i,j}$ such that each entry of the vector is statistically independent, and is sampled from a Gaussian distribution $\mathcal{N}(0,1)$. We also assume that the identities $\veca_{i,j}$ are statistically independent. To generate the results, the NCEE protocol is implemented for $10^4$ packets, wherein a packet is said to be in error if the recovered path $\hat{\vecx}$ is different from the traveled path. Thus, the error rate is defined as the fraction of packets for which the provenance recovery results in error. 

In Fig.~\ref{fig:1}, we present the error rates of the OMP, L-OMP and PL-OMP algorithms when $n = 6, 9$, $h = 3, 5$, and $L=3$. The plots show that L-OMP has a superior performance compared to the vanilla OMP although both methods are inherently path-agnostic. This behaviour is attributed to the fact that L-OMP generates a list of sparse solutions rather than just one, which in turn increases the likelihood of obtaining the path among candidate solutions. Further, we observe that PL-OMP has superior performance as compared to L-OMP, and this is because of two reasons. The first reason is attributed to the path constraint imposed by \eqref{eq:max_projection_path_aware} in every iteration of the OMP-algorithm, The second reason is that in the penultimate iteration, PL-OMP detects and completes a missing-link, ensuring brute-force path formation. In contrast, in the L-OMP algorithm, no path-constraint is imposed in every iteration, and moreover, the $L$ candidates chosen in the last iteration of L-OMP may not ensure path formation. In Fig.~\ref{fig:1}, we also present the error rates for gOMP and L-gOMP as a function of $m$. First, we note that gOMP outperforms OMP, and this behaviour is well known \cite{gOMP}. In addition, we also observe that L-gOMP outperforms gOMP. The reasoning for this observation remains same as noted in the case of OMP. Finally, in Fig.~\ref{fig:3}, we present the error rates of all the proposed algorithms along with that of $\ell_1$ minimization problem (CVX-based implementation) for a small and moderately large network sizes of $n = 7$ and $n = 15$, respectively. The plots show that L-gOMP outperforms both L-OMP and CVX when the column size is small, however, the CVX implementation outperforms the rest when the column size is beyond $24$. Moreover, the behaviour observed with $n = 7$ also holds good with $n = 15$. Although the performance of CVX is superior, these methods cannot be deployed in practice. 

\subsection{Comparison with Bloom filter based Methods}

For comparing the performance of the NCEE protocol with the Bloom filter protocol \cite{HTDSC}, we observe that the entries of $\mathbf{A}_{E}$ must be appropriately quantized so as to keep the same provenance size for both the schemes. One such comparison is presented in Fig. \ref{fig_motivation} by plotting the error rates of the NCEE protocol along with the L-gOMP algorithm at the destination. As highlighted in the caption of Fig. \ref{fig_motivation}, the simulations are conducted for a network of $n = 15$ nodes when $h = 4$. As far as the Bloom filter based method is concerned, we observe that error rates are minimized when the number of hash functions is bounded away from one. Consequently, the processing delay at each relay node is also higher. To generate the absolute delay corresponding to the number of hash functions used at each node, we have implemented these methods on a test bed of XBee devices. With respect to the NCEE protocol, we have constructed $\mathbf{A}_{E}$ from $\{0, 1\}$, and have implemented network coding operations over $\mathbb{R}$. As a result, the number of bits used to represent every symbol in the provenance is $1$ bit, $2$ bits, $2$ bits and $3$ bits across the $1$st hop, $2$nd hop, $3$rd hop and the $4$th hop, respectively, thereby contributing an average of $2$ bits per symbol per hop. With this framework of quantization, we have chosen an appropriate value of $m$ for the NCEE protocol to ensure that the average provenance size is the same as that of the Bloom filter protocol. The plots confirm that the NCEE protocol offers delay profiles in a complementary region compared to \cite{HTDSC}, yet providing low error rates.

\section{Conclusions}

Identifying the requirements of low-latency constraints and varying network topology, we have proposed the NCEE protocol in order to recover the footprint of information flow. For the provenance recovery process of the NCEE protocol, we have proposed a family of path-aware OMP algorithms that capture the path and sparsity constraints in their solutions. Importantly, from the standpoint of novelty, the NCEE protocol caters to lower delay regimes than the Bloom filter based protocols.

\bibliographystyle{IEEEtran}
\bibliography{ref}
\end{document}